\documentclass[journal]{IEEEtran}
\ifCLASSINFOpdf
\else
   \usepackage[dvips]{graphicx}
   \DeclareGraphicsExtensions{.eps}
\fi
\usepackage[cmex10]{amsmath}
\usepackage{url}

\usepackage{amsthm}
\newtheorem{theorem}{Theorem}
\newtheorem{lemma}[theorem]{Lemma}
\newtheorem{corollary}{Corollary}[theorem]

\begin{document}
%
\title{Correcting for Non-Markovian Asymptotic \\ Effects  Using Markovian Representation}
%
%
%

\author{Vitali Volovoi
\thanks{V. Volovoi is an independent consultant in  Alpharetta,
Georgia, USA.}
\thanks {E-mail: vitali@volovoi.com. Version  May 16, 2017.
}
}
%
%

\markboth{Volovoi: Non-Markovian Asymptotic Effects Correction}%
{Volovoi: Non-Markovian Asymptotic Effects Correction}
%


\maketitle
\begin{abstract}
Asymptotic properties of Markov processes, such as steady-state probabilities or the transition rates to  absorbing states, can be efficiently calculated by means of linear algebra even for large-scale problems. This paper discusses the methods for adjusting the parameters of  Markov models to account for non-constant transition rates. In particular, transitions with fixed delays are considered, along with transitions that follow Weibull and lognormal distributions.  Procedures both for  steady-state solutions in the absence of an absorbing state and for transition rates to an absorbing state are provided, and demonstrated in several examples. 
\end{abstract}


 \ifCLASSOPTIONpeerreview
 \begin{center} \bfseries EDICS Category: 3-BBND \end{center}
 \fi
%
\IEEEpeerreviewmaketitle

\section{Introduction}
%
%
%
%
\IEEEPARstart{T}{his}  paper introduces a simple practical  correction to Markov continuous time models to account for variable transition rates between system states. The correction is based on local balance of the outflows from a given node of the state-space model. The correction provides an accurate means of evaluating the relevant asymptotic performance measures of the system. For a detailed review of the methods for evaluating non-Markovian processes, please refer to ~\cite{Distefano2012}.  Therein, the methods are grouped into the following three categories: 

\begin{itemize}
\item {\em Phase-type expansions}: non-exponential transitions  are replaced with  sets of exponentially distributed phases (stages);
\item {\em Markov renewal theory}:  relies on finding points in time when the system is renewed and the memoryless (Markov) property holds. The resulting Markov discrete time process is  thus ``embedded" into the original continuous time process;
\item {\em Supplementary variables}:  the elapsed holding (sojourn) time is explicitly described by supplementary continuous variables associated with each state.
\end{itemize}
The method described in this paper is most closely related to the second category, in particular the two-phase method for finding a steady state solution.  The main  source of distinction is the focus of the current paper on the  asymptotic intensity of local transition flows among states. As a result, there is no need to find a solution for a global embedded discrete Markov process  before addressing the time continuous problem. Instead of trying to solve the problem from scratch, the current approach zooms in on the portion of the model where the asymptotic behavior deviates from that of the corresponding Markov model. 

 For Markov processes with $n$ states and continuous time, the  governing system of (Chapman-Kolmogorov) differential equations  can be written as follows: 
 \begin{equation}
\frac{dP(t)}{dt}=Q\cdot P(t), \quad Q_{ii}=-\sum_{j\neq i}^n{Q_{ji}}
 \label{kolm1}
\end{equation}
Here  $P_i(t)$, $i=1\ldots n$ are the  probabilities of being in state $S_i$: $P_i(t)=Pr\{X=S_i\}$; and $Q$ is the transition rate matrix with the diagonal terms compensating for the off-diagonal terms in each column.

Let us consider a series of simple examples, starting with the simplest renewable ``birth" and ``death" process. The processes under consideration are relevant in multiple domains (see, for example, \cite{Doorn2013}), but for the sake of specificity, we will cast the problem in terms of system reliability, so that there is a system consisting of a single component that can fail with a (constant) failure rate $\lambda$, and repaired with a constant rate $\mu$.
In other words, both transitions are memoryless and follow exponential distributions.
The resulting process has two states, $S_0$ and $S_1$;  it is Markovian, and in the steady state,  the probabilities  $P_0$ and $P_1$ can be obtained by balancing inflow and outflow for either state: $\lambda P_0=\mu P_1$, so $P_0=\lambda/(\lambda+\mu)$.

Let us now consider a system where the holding (sojourn) time at each state follows general distributions $F(t)$ and $G(t)$ for failure and repair, respectively.  To facilitate the comparison, the means of each distribution remain  unchanged: $\bar{F}=1/\lambda$ and $\bar{G}=1/\mu$. The corresponding process is semi-Markovian, as there is no memory of the previous states, but the transition rate from a state does depend on the holding time in that state.

For each of the transitions, the hazard (conditional) transition rate can be defined:
\begin{equation} 
 h_f(t)=f(t)/(1-F(t))
 \label{haza}
\end{equation}
 
Here the subscript $f$ is used to indicate that the transition rate is evaluated for distribution $F$ (similarly, $h_g(t)$ can be defined for distribution $G$).

 Let us explore the analogy between the changes in  state probabilities and the flows of a fluid (the dynamics of both are represented mathematically by similar differential equations). There are two classical specifications of the flow field in continuum mechanics: the Eulerian view considers a specific location and  focuses on how the flow through that location changes with time; in contrast, the Lagrangian view is focused on the dynamics of a fluid ``particle" (parcel). The two views are equivalent, and complementary in terms of providing insights into fluid flow fields~\cite{Lamb1895,LandauLifshitz1987}. 
 
The notion of time $t$ in Eq.~\ref{haza}  can be considered Lagrangian. Indeed, the hazard rate is a function of how long a particular component is being repaired (holding time). In other words,  a ``particle" trajectory is represented. This is closely related  to  a ``local" or component-based view, which can be considered a {\em raison d'etre} for Stochastic Petri Nets (SPNs)~\cite{Marsan1990,Haas2002,Volovoi2015}.  In SPNs, the individual components of a system are explicitly represented, and  their histories can be tracked.
 
 The ``global" Markovian representation of a state space  gravitates  toward the Eulerian view: the dynamics of individual ``particles" are not represented---or rather they all behave in identical manner, and their individual pasts are irrelevant. Semi-Markov models tilt the view toward the Lagrangian perspective, as they effectively embed a ``particle" representation for each individual state by introducing the dependence of a conditional transitional rate, Eq.~\ref{haza}, on the holding time in a state.  However, one can revert to the Eulerian viewpoint and consider the intensities of the  flows of aggregated particles as a function of global time, as well as in a steady state. This is the approach taken in this paper.
 
 For the considered single-part example, the steady-state intensity of such aggregated transitions will remain the inverse of the means ({\em i.e.,} $1/\bar{F}$ and $1/\bar{G}$, respectively), regardless of the type of distribution used for each transition. In other words,  a Markovian approximation does not introduce any errors in estimating an aggregate measure, such as the chances that the system will be operational in a steady state; however, this is not always the case, as illustrated next.
 
Let us consider a two-part redundant system,  with the state representation shown in Figure~\ref{markov2redundant}. The two parts are identical, so there are only three states.  
  \begin{figure}[!t]
\centering
\includegraphics[width=2.5in]{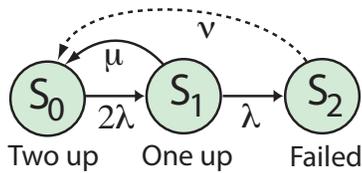}
\caption{Markov model of a two-part redundant system. The dashed line represents an artificial transition to avoid an absorbing state}
\label{markov2redundant}
\end{figure}
The transition rate matrix has the following form:
\begin{equation}
\quad Q=\left(
\begin{array}{ccc}
-2\lambda&\mu&0\\
 2\lambda&-\lambda-\mu&0\\
0&\lambda&0\\
\end{array}
\right)
\label{kolm3}
\end{equation} 
The focus is on  the frequency of system failures, {\em i.e.,} the hazard rate of the system, as defined by Eq.~\ref{haza}  and interpreted at the system level ({\em i.e.,} he transition from any non-failed state to a failed state is considered). See~\cite{Doorn2013} for a comprehensive review of the general class of problems with quasi-stationary distributions, the phenomenon that first was explicitly identified by A. M. Yaglom~\cite{Yaglom1947}. This hazard (decay) rate  is of significant practical value~\cite{Volovoi2017mmr}, and there are two main approaches to its calculation for Markovian systems.

 The first method considers an artificial transition (depicted in Figure~\ref{markov2redundant} by a dashed line) in order to convert the considered process with an absorbing state into a renewal process. This implies that the last column in Eq.~\ref{kolm3} will have two non-zero entries: $\nu$ and  $-\nu$ in the first and the last row, respectively. 
 
 After finding a steady-state solution, one can evaluate conditional probabilities $P_0$ and $P_1$ given that the system has not failed ({\em i.e.,}  $\hat{P}_i=Pr\{X=S_i | X \notin F\}, i=0,1$; in this case there is only one failed state $S_2$). Using the balance equation for the $S_1$, we observe that $(\lambda+\mu)\hat{P}_1=2\lambda \hat{P}_0$.  Let us consider $\lambda=\mu=1$, then $\hat{P}_0=\hat{P}_1=0.5$, so that the  system failure hazard rate is 
  \begin{equation} 
 h=\hat{P}_1 \lambda=\frac{2\lambda^2}{3\lambda+\mu}=0.5
  \label{haza1}
\end{equation}
Note that the rate of artificial transition $\nu$ does not affect the solution.  This independence  is not accidental, as shown in Lemma 3 toward the end of this paper.
The second method instead relies on the Perron-Frobenius theorem to determine  
 the systems' hazard rate. 
 \subsection{Perron-Frobenius eigenvalue}
 The key assertion underpinning this method has been used in~\cite{Darroch1965} for discrete Markov chains (see also ~\cite{Boussemart2001} for a practical application),  and here it is demonstrated directly for continuous Markov processes (originally derived in~\cite{Darroch1967}). Let us consider a system where all non-absorbing states belong to the same communicative class ({\em i.e.,} any state can be reached from any other)~\cite{Haverkort2001}. Then the following holds:
 \begin{lemma}
 The asymptotic rate of transition to an absorbing state is the absolute value of the smallest negative eigenvalue of the transition matrix $Q$.
\end{lemma}
\begin{proof}
  Let us consider a general case with $n$ ``up" (non-absorbing) states $U_k, k=1\ldots n-1$, and a single failed (absorbing) $n$-th state $F$. Combining several absorbing states into a single absorbing state does not lead to a loss of generality, since  the relevant dynamic of the system only concerns the ``up" states. After this dynamic is described by evaluating the temporal history of the probability of each ``up"  state, the rates of transition to each of the failed states can be segregated if desired.
 
Let us recover the transition matrix from time $t$ to $t+\delta$: $P(t+\delta)=\Pi P(t)$ (here   $\delta$ is a small time step). Using Eq.~\ref{kolm1}, it can be expressed as 
$\Pi=\delta Q+I$, where $I$ is an identity matrix, and $\Pi$ is a stochastic matrix with the last column representing the absorbing state. This last column has the diagonal term $1$ and all zeros for the  off-diagonal terms. As a result, $\Pi$  is a block diagonal matrix  with the first block $\tilde{\Pi}$  of size $n-1$  corresponding to the ``up" states and the last block of size $1$ corresponding to the absorbing (``down") state. 

 In accordance with the Perron-Frobenius theorem for $\tilde{\Pi}$, there is a unique and distinct largest eigenvalue $\tilde{k}<1$.  In addition, there is an  eigenvalue of $1$ stemming from the absorbing state.  Converting these eigenvalues to the corresponding eigenvalues for matrix $Q$, we can conclude that its two largest eigenvalues are $0$ for the absorbing state and the negative Perron-Frobenius  eigenvalue  $-k=\tilde{k}-1$  (here $k>0$), with the rest of eigenvalues being all negative and strictly less than $-k$.

As a result, we can represent the probability of an ``up" state as follows:
\begin{equation}
U_i (t) =A c_i+B V_i e^{-kt}+ \ldots
\label{pf1}
\end{equation}
and the ``down"  state 
\begin{equation}
F(t) =A c_{n}+B v_{n} e^{-kt}+ \ldots
\label{pf2}
\end{equation}

Here the dots represent the contributions from the rest of the eigenvalues that are all negative with absolute values larger than $k$, and where $c_i$ and $v_i$ are components of the first two right eigenvectors (corresponding to the  first two largest eigenvalues). These contributions, represented by dots, can be neglected for sufficiently large instances of time. Taking these expressions to the limit $+\infty$, and given that  $F$ is an absorbing state (so that its probability  tends to unity), yields
\begin{equation}
\lim_{t\to+\infty} Pr\{F\}(t)=1-B v_{n}e^{-kt}
\label{pf3}
\end{equation}
 and  similarly, 
 \begin{equation}
\lim_{t\to+\infty} Pr\{U_i\}(t)=B v_i e^{-kt}
\label{pf4}
\end{equation}
  Recalling the definition of the hazard rate, Eq.~\ref{haza1}, and differentiating Eq.~\ref{pf3}  yields
   \begin{equation}
h(t)=\frac{dF/dt}{(1-F(t))}=\frac{k B v_{n}e^{-kt}}{B v_{n}e^{-kt}}=k
\label{pf5}
\end{equation}
\end{proof}
There are three corollaries that follow from Eq.~\ref{pf4}: 
\begin{corollary}
Non-absorbing states reach a quasi-steady state. \end{corollary}
Indeed, the conditional probabilities of each ``up" state $U_i$ asymptotically do not depend on time, because they  are simply the corresponding components of the Perron-Frobenius eigenvector
  $\hat{P}_i=Pr\{X=U_i |X\notin F\}=v_i $, where the right eigenvectors are normalized, so that  $\sum_{i=1}^{n-1}{v_i}=1$. 
  \begin{corollary}
  The probability of being in each non-absorbing states decays as $e^{-kt} $.\end{corollary}
  Another way to represent the eigenvalue $k$ can be also useful:
 
   \begin{corollary}
  The hazard rate is a weighted sum of direct transition rates to the absorbing state, with the weighting provided by the conditional probabilities  $\hat{P}_i$:
  \begin{equation}
k=\sum_{1}^{n-1}{v_i Q_{ni}}=\sum_{1}^{n-1}{\hat{P}_i Q_{ni}}
\label{pf66}
\end{equation} 
  \end{corollary}
  To demonstrate that this corollary is true,  Eqs.~\ref{pf3} and \ref{pf4} are summed together as probabilities that add up to unity, which yields $\sum_{i=1}^n{v_i}=0$. Given the chosen normalization in the Corollary 2,  $v_n=-1$, Eq.~\ref{pf66} is simply the last row of the matrix expression for the eigenvalue $-k$: $Q\mathbf{v}=-k\mathbf{v}$ (where $\mathbf{v}$ is the vector with components $v_i$).
  
  Returning to  the two-part redundant system example (with the state space depicted in Figure~\ref{markov2redundant} and matrix $Q$ provided in Eq.~\ref{kolm3}), and applying Lemma 1,   one can write the characteristic equation of the block  matrix  for the first two states:
\begin{equation}
k^2+k(3\lambda+\mu)+2\lambda^2=0
\label{characteristic3}
\end{equation}
The absolute value of smallest (negative) root for Eq.~\ref{characteristic3} is
\begin{equation}
k_1=\frac{3\lambda+\mu-\sqrt{(3\lambda+\mu)^2-8\lambda^2}}{2}
\label{characteristic4}
\end{equation}

Evaluating Eq.~\ref{characteristic4} for specific numerical values,  $h_{pf}=k_1=0.58579$ is noticeably  higher than the value obtained using  the renewal process earlier (see Eq.~\ref{haza1}). The difference becomes significantly smaller when the repair rates are larger than the failure rates ({\em i.e.,} the non-dimensional parameter $\rho=\lambda/\mu\ll 1$); see Figure~\ref{pfvsrenewal}. 
\begin{figure}[!t]
\centering
\includegraphics[width=3in]{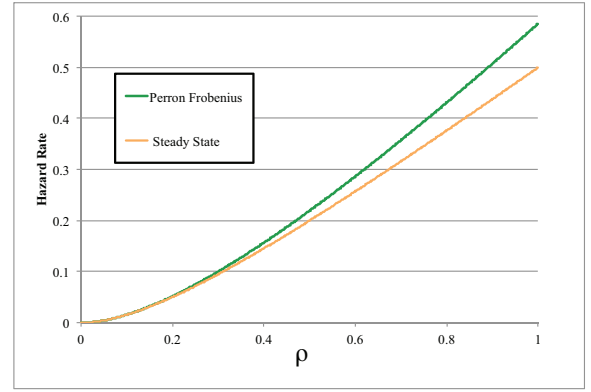}
\caption{Comparison between hazard rates based on the Perron-Frobenius eigenvalues vs. renewal process}
\label{pfvsrenewal}
\end{figure}
The difference is explained by the influence of the transient behavior: the renewal approach effectively averages over  the renewal  time segment, which starts with a fully functional system, as opposed to the purely asymptotic value provided by the Perron-Frobenius eigenvalue. For Markov processes, the size of the transient effect is related to the relative difference between the Perron-Frobenius eigenvalue and the next largest eigenvalue  modulus or SLEM~\cite{Boyd2004}.
\subsection{Semi-Markov effect}

In contrast  to a single-component system,  for a two-part redundant system  the asymptotic values characterizing the system will depend on the types of distributions assigned to the repair transitions (and not only on the mean values). For example, let us consider repairs that are completed after a  fixed delay. Keeping the same mean value for the repair, the hazard rate will increase compared to the exponential distribution.

Figure~\ref{fixedDelay} depicts the corresponding hazard rate as calculated by two ``brute force" methods: Monte Carlo simulation using $100$ million replications, and a forward-marching finite difference scheme with a step of $6\times10^{-4}$ and $10,000$ steps. Here the holding (sojourn) time distribution for repair is a part of the state description. 

It may be observed that the Monte Carlo simulation has more noise toward the end of the simulation. This is expected given the fact that, as simulation time progresses, there are fewer sample paths that correspond to the non-failing state ({\em cf.}  \ref{pf4}) and the accuracy of the method decreases. Nevertheless, the trend is clearly observed---the initial oscillations die out, and a steady state is reached. Both the presence of oscillations and the associated longer transient duration (as compared to exponentially distributed repairs) are artifacts of the repairs with fixed delays. The hazard rates time averaged over  $t\in [4,6]$  for the finite difference and Monte Carlo simulations are  $0.62513$ and $0.6265$, respectively. We expect the time-difference method to be more accurate, but it is clear that the true value is significantly higher than the one observed for the exponential repair ($0.58579$). 
 \begin{figure}
\begin{center}
\includegraphics[keepaspectratio=true,width=3.0in]{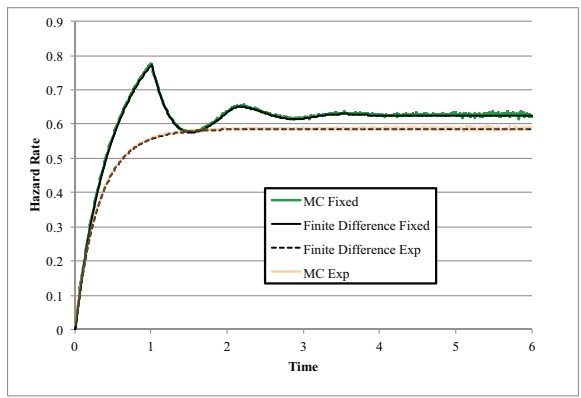}
\end{center}
\caption{Comparing hazard rates for a two-part  system with repairs occurring after a fixed delay of $\tau=1$. \label{fixedDelay}}
\end{figure}

The next three sections of this paper will address the  corrections for the following cases: 
a steady-state solution for the semi-Markovian model,  a hazard rate in the presence of an absorbing state; and finally, some of the cases where abandoning the Markovian assumption leads to non-Markovian (rather than semi-Markovian) processes.

\section{Steady-state correction for semi-Markov processes}
Let us consider a generic node of a state-space representation (see Figure~\ref{genericNode}). First, we note that all inputs from the states $I_1\ldots I_k$ can be combined into a single inflow: in fact, the magnitude of this inflow is not important (it will be normalized later). What is relevant is that this inflow has a constant rate in a steady state.  If there is only a single output to node $O_1$, then (regardless of the particular distribution for that transition) the only relevant aspect of this transition is the overall intensity, or (equivalently), the mean value of the distribution. So, in this case the substitution of a non-exponential distribution for an exponential one is trivial: only the mean value needs to be matched.  
\begin{figure}[!t]
\centering
\includegraphics[width=2.2in]{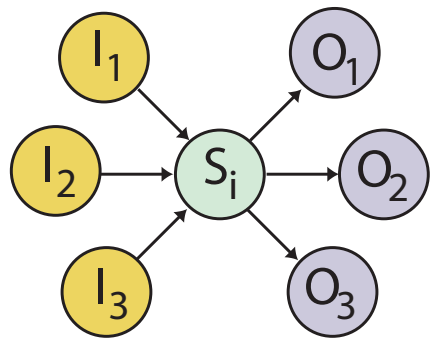}
\caption{Generic node representing a single state in a state-space representation}
\label{genericNode}
\end{figure}

If there are several exponential distributions, their effect can be combined by summing up their respective transition rates. For example, let us assume that for the node shown in Figure~\ref{genericNode},  transitions to the states $O_1$ and $O_2$ follow exponential distributions with the rates $\lambda_1$ and $\lambda_2$, respectively. Then the effect of the outflow can be represented by a combined exponential distribution with the intensity $\lambda=\lambda_1+\lambda_2$. 

Let us consider a general case and focus on a transition to a particular state that  follows a general distribution, with a cumulative distribution $G(t)$ and probability density function $g(t)$.  All other outflows are combined into an equivalent single distribution
\begin{equation}
F(t)=1-\prod_{j=1}^m{(1-F_j(t))}
\label{combinedFlows}
\end{equation}

From this perspective, there is a  ``race"  between the outflow transitions. In that respect, there is a similarity to the competing risks problem considered in~\cite{VolovoiVega2012}.  


Given that  inflow into the node is constant, the key consideration is the  impact of the outflow from the node on the distribution of the holding time $t$ of the outflow. Specifically,  at time $t_0$, the probability that the holding time was $t$ ({\em i.e.,} the last state transition occurred at time $t_0-t$) is  decreasing as  $t$ increases: as more time elapses there are more  chances  for the outflow to occur.
More precisely, for a given $t$, this value is $R(t)=A(1-G(t))(1-F(t))$.  Here $A$ is a normalization constant that can be determined by integrating this expression over time:

\begin{equation}
\int_{0}^{+\infty}{(1-G(t))(1-F(t))dt}=1/A
\label{age}
\end{equation}
Once this constant is determined, the overall transition rate  can be evaluated as the weighted average (based on the age distribution):

\begin{eqnarray}
\hat{\mu}=A\int_{0}^{+\infty}{(1-G(t))(1-F(t))h_g(t) dt}=\nonumber\\
=A\int_{0}^{+\infty}{(1-F(t))g(t) dt}=A\gamma
\label{haza4}
\end{eqnarray}
Here $h_g(t)$ is the hazard rate for distribution $G$, as provided by Eq.~\ref{haza} and  the ``winning race ratio"  $\gamma$  represents the chances that $g$ will ``win the race" with the competing outflow ({\em cf.}~\cite{VolovoiVega2012} where a similar ratio was analyzed for small time scales). As discussed below, this ratio also corresponds to the appropriate term in the embedded discrete time transition matrix.

First we note that if there is only one outflow  transition from the node, then $F(t)\equiv 0$, and from Eq.~\ref{haza4} we can readily observe that $h=A$. From Eq.~\ref{age} we can use the integration by parts to obtain
\begin{eqnarray}
1/A=\int_{0}^{+\infty}{(1-G(t))dt}=-\int_{0}^{+\infty}{td(1-G(t))}\nonumber\\
=\int_{0}^{+\infty}{tg(t)dt}=\bar{G}
\label{ageA}
\end{eqnarray}
This confirms the assertion made in the beginning of the paper that for a system with one part, the transition rates are fully determined by the mean value of the distribution (so that the transition rate is its inverse).  Furthermore, a similar assertion can be made for all the nodes without a ``race" between different outflows (or if all outflows are aggregated together to measure the entire outflow).

Next, we can make some simplifications when $F$ is an exponential distribution. 
First, let us rewrite  Eq.~\ref{haza4}: 
\begin{eqnarray}
\hat{\mu}=A\int_{0}^{+\infty}{e^{-\lambda t}g(t) dt}=A\gamma
\label{haza4a}
\end{eqnarray}
and $\gamma$ is the only constant that is needed to evaluate the transition rate. Indeed,
Eq.~\ref{age} can be written as follows:
\begin{eqnarray}
1/A=\int_{0}^{+\infty}{e^{-\lambda t}(1-G(t))dt}=\nonumber\\
=-\frac{1}{\lambda}\int_{0}^{+\infty}{(1-G(t))de^{-\lambda t}}=\nonumber\\
=\frac{1}{\lambda}\left[1+\int_{0}^{+\infty}{e^{-\lambda t}d(1-G(t))}\right]=\frac{1-\gamma}{\lambda}
\label{ageB}
\end{eqnarray}
Combining Eqs.~\ref{haza4} and~\ref{ageB}, we obtain the expression for the transition rate competing with an exponential outflow:
\begin{equation}
\hat{\mu}=
\frac{\lambda}
{1/\gamma-1}
\label{trans}
\end{equation}
Important special cases can be noted: first,  if $G$ is an exponential distribution, then $\gamma=\mu/(\lambda+\mu)$, and from Eq.~\ref{trans} we can observe that no correction is needed to the transition rate ($\hat{\mu}=\mu$). Finally, for a fixed delay $\tau$, we observe that $\gamma=e^{-\lambda\tau}$, so the  corrected transition rate has the following form:
\begin{eqnarray}
\hat{\mu}=\frac{\lambda}{e^{\lambda \tau}-1}
\label{haza5}
\end{eqnarray}

Let us return to the example with two redundant components, and consider a constant failure rate $\lambda=1$ and a fixed repair delay $\tau=1$. The effective repair rate will be $\hat{\mu}=0.58198$, which, substituting it into  Eq.~\ref{haza1}, translates into the hazard rate $h_f=0.55835$. Verifying  this result using Monte Carlo simulation up to time $1000$  with  1 million replications provides the value of $0.55829$. In summary, when comparing to scenario where the repairs follow an exponential distribution with the same mean, the effective repair rate $\mu$ drops  from $1$ to $0.58198$, which causes an increase of the hazard rate from $0.5$ to $0.55835$.

For steady-state semi-Markov processes, the described procedure is equivalent to the well-known two-step evaluation using renewal processes~\cite{Haverkort2001}: at the first step, steady-state probabilities  $\pi_i$ for each state of the embedded discrete Markov chain are obtained; at the second step, those probabilities are weighted by the mean holding times at each state $m_i$:
\begin{eqnarray}
P_i=\frac{\pi_im_i}{\sum_{j=1}^n{\pi_jm_j}}
\label{renewal}
\end{eqnarray}
The following statement holds:
  \begin{lemma}
 Finding a steady-state solution for a continuous Markov process with all non-exponential transitions being corrected by using Eqs.~\ref{age}, \ref{haza4} to calculate equivalent transition rates is equivalent to solving the two-step renewal process~\cite{Haverkort2001}.
\end{lemma}
\begin{proof}
To show the equivalency, we demonstrate that both transitional probabilities of the embedded discrete Markov chain and the mean holding times are preserved by the procedure developed in this paper. Indeed, let us consider a state-space node $S_i$ (see Figure.~\ref{genericNode}), with output transitions $T_1\ldots T_l$ following general distributions $G_1\ldots G_l$, and focus on a single transition $T_{w}$ to a particular output state $O_{w}$. This transition
 corresponds to the globally numbered state $S_j$ with the rest of the transitions combined into a single distribution $F$  (Eq.~\ref{combinedFlows}). Then,  one can observe that   $\gamma$ introduced in Eq.~\ref{haza4} by definition represents the corresponding term of the embedded discrete transition matrix $K_{ij}$. 
 
  Further combining $G_{w}$ with $F$ into a single distribution $G$ for all outflows from the $S_i$ and invoking Eq.~\ref{ageA}, yields $1/A=m_i$.
\end{proof}

While the end result is the same when both methods are applied to finding the steady-state solution of a  semi-Markov process, the method described in this paper directly focuses on the local deviation of the transition flow intensity from an exponential transition. As a result, only a  single global problem needs to be solved, and the impact of individual transition distributions is more transparent before the global problem is solved (one can immediately observe whether the equivalent transition rate is smaller or larger as compared to the value implied by the mean). Furthermore, the flow intensity perspective provides a stepping stone for finding asymptotic rates to absorbing states for semi-Markov processes and solving non-Markovian problems, as described next.

\section{Asymptotic hazard rate for semi-Markov process}
In the presence of an absorbing state, a similar  local  flow  balancing can be applied as well. To that end, a Eulerian view of the semi-Markov process  is adopted and an equivalent Markov process is considered (the convergence of the process to a quasi-stationary state can be assumed~\cite{Doorn2013}).  Invoking Corollary 1.1  and the presence of a quasi-steady state implies that the relative probabilities for inputs to  a generic node (as depicted in Figure~\ref{genericNode}) do not change with time, so all inflows can be effectively combined into a single inflow. In contrast to the steady-state scenario, there an additional  decay term $e^{-kt}$ applied to all  the input states (see Corollary 1.2). Therefore,  Eq.~\ref{age} must be adjusted to provide relatively more weight to older ages (recall that  $t$ represents the backward time from some instant $t_0$, so the $e^{kt}$ is added):
\begin{equation}
\int_{0}^{+\infty}{e^{kt}(1-G(t))(1-F(t))dt}=1/A
\label{age5}
\end{equation}
In other words, relatively speaking, there will be more chances for greater holding time, since the inflow was stronger in the past.
As a result,  Eq.~\ref{haza4} changes to
\begin{eqnarray}
\hat{\mu}=A\int_{0}^{+\infty}{e^{kt}(1-F(t))g(t) dt}
\label{haza6}
\end{eqnarray}

When the transition in question is competing with another transition that follows  an exponential distribution with the rate $\lambda$,   we can use similar formulae to those of the steady-state case, while replacing $\lambda$ with $\lambda-k$:
\begin{eqnarray}
\gamma_a=\int_{0}^{+\infty}{e^{-(\lambda-k) t}g(t) dt}\label{gamma2}\\
\hat{\mu}=\frac{\lambda-k}{1/\gamma_a-1}
\label{mu2}
\end{eqnarray}

One interesting implication of Eqs.~\ref{age5},~\ref{haza6} is that even in the absence of a race among the outflows from a state, the effective outflow rate deviates from the exponential equivalent, and the only time when there is no deviation is when there is a competing exponential outflow with intensity $\lambda=k$.

The described procedure  appears to be relying  on a circular argument: $k$ is obtained as a Perron-Frobenius eigenvalue of the transition matrix, which requires the knowledge of all transition rates for matrix $Q$. Therefore, an iterative procedure is implemented: starting with a guess (for example, from using uncorrected exponential values), convergence is achieved  in all the examples considered below, as well as for a more involved example considered in~\cite{Volovoi2017mmr}.
Moreover, the following observation provides the grounds for conjecture that the convergence can be achieved under fairly general conditions. 

Let us consider a $j$th step of an iteration with the value of the Perron-Frobenius eigenvalue $k(j)$. This value can be substituted into Eqs.~\ref{haza6}, \ref{age5} to solve for new values of all  transition rates that require the adjustments  $\hat{\mu}_{i}(j+1)$.  After this, the new Perron-Frobenius eigenvalue $k(j+1)$ is calculated. In other words, there is a sequence of two alternating mappings:
\begin{eqnarray}
k(j)\mapsto \hat{\mu}_{i}(j+1)\mapsto k(j+1)\ldots
\label{mapping}
\end{eqnarray}
To ensure that the iterations converge, it is critical that the mappings contract, so that the differences between the successive iterations shrink sufficiently fast. The first mapping (as described by  Eqs.~\ref{haza6}, \ref{age5}) is ``local" ({\em i.e.,} related to a single transition and determined by the combination of the distributions for the involved outflow transitions), while the second is ``global" (as determined by the  system configuration and expressed as the Perron-Frobenius eigenvalue).  

In this context, a general three-state system (depicted in Figure~\ref{extremeSystem})  should  provide enough flexibility to explore the different scenarios of the second mapping.  For example, let us  design a configuration where changes in the transitions result in the largest relative change in the hazard rate. For specificity, let us consider a situation where the objective  is to increase the hazard rate ($k\uparrow$). Recalling Corollary 1.3, we can observe that this can be achieved by increasing the intensity of the ``edge" transitions (i.e., directly into the absorbing state, $ e_1\uparrow$, $e_2\uparrow$), and by increasing the conditional probability of the nodes with the largest intensity of the ``edge" transitions (assuming $e_2>e_1$ for specificity,  $i_1\uparrow$, $i_2\downarrow$). 

Intuitively, introducing more states simply dilutes the effect. Indeed, let us order  the states  based on the intensity of their respective edge outflows and  compare the result to a three-state system that contains only  the safest and the least safe states. Both actions that impact the hazard rate ({\em i.e.,} the change in the intensity of the edge transitions, and shifting the conditional probabilities within the non-absorbing states; see Eq.~\ref{pf66}) have more direct (and therefore more pronounced) effects in the three-state system. Therefore, the sensitivity of the hazard rate with respect to the changes in the transition rates will be dampened in a system with intermediate states, and the overall convergence rate will increase. In other words, testing the convergence for a particular combination of transition distributions on the three-state system provides a good indication of the convergence of the described procedure for a larger system. 
\begin{figure}[!t]
\centering
\includegraphics[width=2.5in]{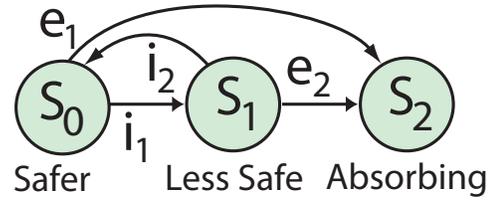}
\caption{A general three-state system}
\label{extremeSystem}
\end{figure}
 
 The extreme scenario of the combination of changes for the transitions described above requires a deliberate choice of different types of distributions for each transition. For transitions with fixed delays,  Eqs.~\ref{gamma2}, \ref{mu2}) simplify to
\begin{eqnarray}
\hat{\mu}=\frac{\lambda-k}{e^{\tau \left(\lambda-k\right)}-1}
\label{haza7}
\end{eqnarray}

This enables us to construct an approximation of the extreme scenario described above:  from Eq.~\ref{haza7} it is clear that increasing $k$ will increase the effective rate of a fixed transition. To exploit this property, we will use fixed delays for transitions $e_2$ (to increase the edge outflow from the least safe place) and for $i_1$ to skew the balance of the non-absorbing states toward the least safe place. The two remaining transitions follow exponential distributions. Let us have fixed delays $\tau_1=\tau_2=1$ and the exponential rates $\lambda_1=\lambda_2=0.1$. The equivalent intensity for both fixed-delay transitions  $\hat{\mu}_1=\hat{\mu}_1=1.913$ with the resulting hazard rate $k=1.5756$. Both transition rates and the system hazard rates are almost double their exponential counterpart ($k_{exp}=0.78377$), so this is quite an extreme scenario in terms of the size of the non-Markovian effects. Expectedly, the rate of convergence is relatively slow as well (see Fig.~\ref{convergence}, the green curve).
\begin{figure}[!t]
\centering
\includegraphics[width=3.0in]{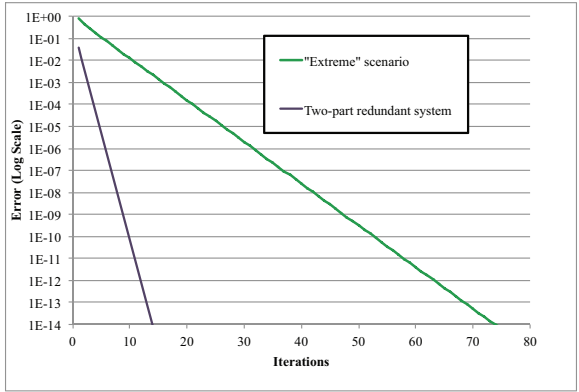}
\caption{Comparison of rates of convergence for the ``extreme" scenario with a pronounced non-Markovian effect vs. for the two-part redundant system}
\label{convergence}
\end{figure}



Next, we return to the two-part redundant system and  apply Eq.~\ref{haza7} to find the equivalent repair rate. For the considered values of $\lambda=1/\tau=1$,  the calculated value for $\hat{\mu}=0.82427$ and $k=0.62518$. This matches  the finite difference value $k=0.62513$ quite well. In Figure~\ref{convergence}, the purple curve shows both the faster rate of convergence and the smaller non-Markovian effect as compared to the ``extreme" scenario. The vertical axis in Fig.~\ref{convergence} is logarithmic, so that the absolute deviation from the converged solution is shown. However, it is worth noting that, in contrast to the ``extreme" scenario, the deviation alternates in sign. This makes physical sense: an increase in the repair rate ($\hat{\mu}\uparrow$) decreases the hazard rate ($k\downarrow$) and vice versa.

Next, we employ Eqs.~\ref{gamma2},~\ref{mu2} to explore repairs that follow Weibull and lognormal distributions. To facilitate the comparison, mean values are kept at unity, and the shape parameter $\beta$ is varied for Weibull distribution, while a ratio of variation to the square of a mean, called squared coefficient of variation (SCV), is used as a varying parameter for lognormal distribution. 

Figure~\ref{weibullSense} shows the results for repairs  that follow a Weibull distribution with a shape parameter $0.3\leq \beta\leq 5$.  For small values of $\beta$, the ``winning race" ratio $\gamma_a$ (Eq.~\ref{gamma2}) is very high, so that the equivalent $\hat{\mu}$ is significantly higher than $1$. In fact,  for $\beta<0.3$, the values are even higher (which is the reason they are not shown on the chart).  Large values of $\hat{\mu}$ make the system less likely to stay in a degraded state, so the system hazard rate is smaller. This effect is attenuated as $\beta$ increases, with $\beta=1$ resulting in the values previously discussed for exponential repair (with $\hat{\mu}=1$, $\gamma_a=0.70711$, and $h_a=0.58579$). For larger $\beta$,  $\hat{\mu}$ decreases slightly, eventually leveling off, with a corresponding slight increase in the system hazard rate. The results are consistent with the trends discussed in~\cite{VolovoiVega2012}.

In Figure~\ref{logSense} the results are shown for repair following lognormal distribution.  For very small  SCV the lognormal distribution mimics a fixed delay repair, so as expected, the values on the left side of the figure replicate the values obtained earlier for the fixed delay. As SCV increases, the  $\gamma_a$ increases slightly, which results in an increase of the effective rate of repair $\hat{\mu}$ and the reduction of the system hazard rate. When SCV=1 ({\em i.e.,} the same value as for an exponential distribution) the results are close to those for the exponential distribution but are  slightly off: $\gamma_a=0.70425$, $\hat{\mu}=0.97273$,  and $h_a=0.59150$. For larger values of SCV, the hazard rate continues to be reduced, so for $\text{SCV}=5$ the values are $\gamma_a=0.73608$,  $\hat{\mu}=1.31689$, and $h_a=0.52784$. It is interesting to note that the trend is inverse to that observed in queues where reduction of SCV is beneficial for queue efficiency~\cite{Whitt1993}.

\begin{figure}[!t]
\centering
\includegraphics[width=3in]{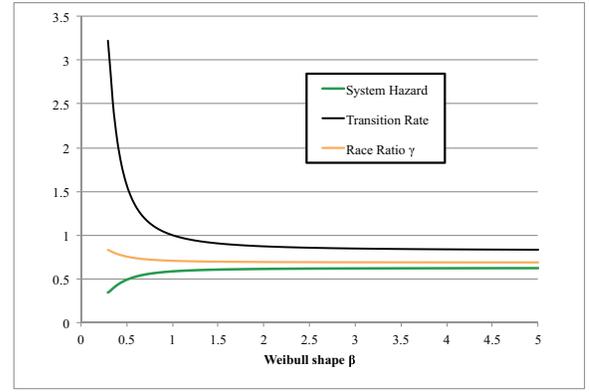}
\caption{Sensitivity to the shape parameter for Weibull repair}
\label{weibullSense}
\end{figure}

\begin{figure}[!t]
\centering
\includegraphics[width=3in]{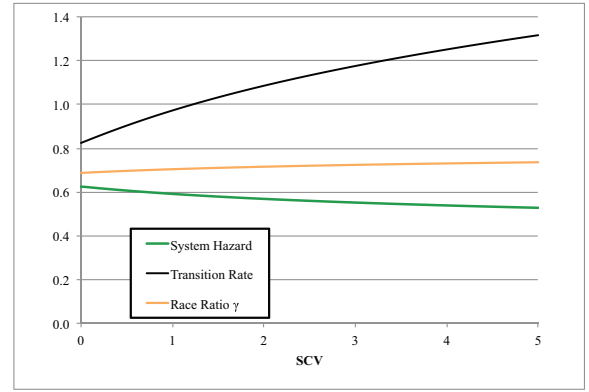}
\caption{Sensitivity of lognormal repair to SCV}
\label{logSense}
\end{figure}

\section{Non-Markovian scenarios}
The examples considered above demonstrate the efficiency  of the approach to a range of semi-Markovian models (see also~\cite{Volovoi2017mmr} for a more involved example). However,  in many practical situations, transitioning from a Markov to a semi-Markov representation requires additional efforts due to the fact that  not all the states are regeneration points of the modeled  process. 

To demonstrate the resulting challenges, let us consider a simple variation of the two-part redundant system example by introducing a  single repair server. When both components have failed, the component that has failed earlier continues to be repaired, and upon its repair the second repair commences. The performance of the system is measured in terms of the steady-state probabilities for operating states ({\em i.e.,} the system availability). The corresponding Markov model is shown in Fig.~\ref{markovSingleRepair}.  If repairs follow exponential distributions, then the solution is straightforward, with $\mu_1=\mu_2=\mu$.  Instead, let us consider repairs of fixed duration $\tau$. 

\begin{figure}[!t]
\centering
\includegraphics[width=3in]{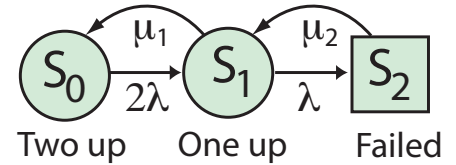}
\caption{Markov model for single repair}
\label{markovSingleRepair}
\end{figure}

The resulting model is not semi-Markovian in the following sense: when the system arrives at the $S_2$ state, it implies that in the $S_1$ state  the repair ``lost the race" to the failure of another component. In other words, some time of the repair has already elapsed, so there is no regeneration point in the $S_2$ state: the remaining time for repair depends not only on the holding time in $S_2$, but also on how much time was previously spent in $S_1$. Following~\cite{Distefano2012}, the corresponding state is depicted as a square (as opposed to a circle) to differentiate  states where no regeneration of the process occurs.

In contrast,  $S_1$  is a regeneration point for the repair: whether the system arrives to the state from  $S_0$ or $S_2$, the repair will start from the beginning, so $\mu_1$ can be calculated using Eq.~\ref{haza5}.  One way to calculate the equivalent rate $\mu_2$ is by noticing that while from the system-level perspective the $S_1$  and $S_2$ are distinct, there are no changes occurring to the repair process itself (since no preemption of the repair is considered). In other words, the overall repair rate remains $\mu$:
\begin{eqnarray}
\hat{P}_1\mu_1+\hat{P}_2\mu_2=\mu
\label{haza8}
\end{eqnarray}
Here $\hat{P}_i=Pr\left\{X=S_i|X\neq S_0\right\}$ are conditional probabilities. We can solve this equation to find out the transition rate:
\begin{eqnarray}
\mu_2=\frac{\mu\lambda}{\lambda+\mu_1-\mu}
\label{haza9}
\end{eqnarray}
The resulting states are $P_1=0.53391$ and $P_2=0.31072$. From simulation the values are $P_1=0.53390$ and $P_2=0.31072$ (1 million replications, with time averaging for times $t\in[100,1000]$). 

This method has obvious scaling limitations (that is, knowing the overall rate will not be sufficient to recover individual rates if there are more than two states), so it is useful to look for a more direct alternative that accounts for the holding time in the previous state. To that end, it is easier to deal with the mean time to repair, $m_2=1/\mu_2$, and directly integrate the remaining time to repair based on the intensity of the transition into the $S_2$ state:

\begin{eqnarray}
m_2=\frac{1}{F(\tau)}\int_0^\tau{\lambda e^{-\lambda t}(\tau-t)dt}=\frac{\tau}{F(\tau)}-\frac{1}{\lambda}
\label{mttr}
\end{eqnarray}
It is easy to check that Eqs.~\ref{mttr} and \ref{haza8} are equivalent.

Before turning to the last example, let us introduce a  simple but useful Lemma:
\begin{lemma}
 In a steady state, if a state $S_i$ has only a single outflow, the intensity $\mu_i$ of this outflow does not impact the conditional probabilities of all other states  $\hat{P}_j=Pr\left\{X=S_j|X\neq S_i\right\}.$ \end{lemma}
\begin{proof}
The linear system of equations for the steady state $Q\mathbf{P}=0$ contains only combined terms $\mu_i P_i$, which implies that changing $\mu_i$ will simply scale $P_i$ (and therefore impact the overall normalization for all probabilities), but it will not impact the relative magnitudes of the probabilities for all other states. \end{proof}
This lemma explains the independence of the solution for an asymptotic hazard rate from the rate of an artificial transition $\nu$ introduced to avoid an absorbent state for the two-part redundant system mentioned earlier.  

\begin{figure}[!t]
\centering
\includegraphics[width=3in]{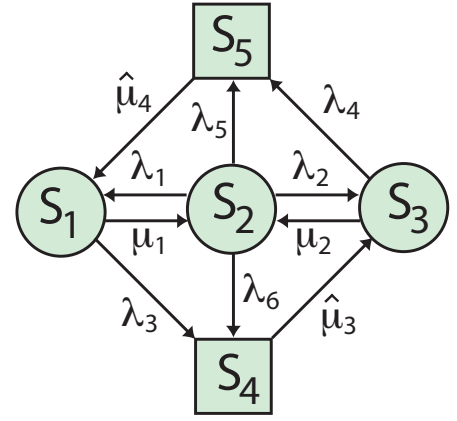}
\caption{State space diagram for a redundant system with two distinct parts and common failure mode}
\label{distefanoD}
\end{figure}

Finally, let us consider an example from~\cite{Distefano2012} of a two-part redundant system that has a common failure mode and  single repair crew, and where the two parts are distinct (a similar system has been also studied in~\cite{Pijnenburg1993}). The state diagram is shown in Figure~\ref{distefanoD}. The numeration of the state space and the numerical values of parameters follow~\cite{Distefano2012}. State $S_2$ (in the center of the diagram) corresponds to the fully operational state. Single failures of the first and second part  trigger transitions to the states $S_1$ and  $S_3$,  with the rates $\lambda_1=(1-q)\delta_A\lambda_A$ and $\lambda_2=(1-q)\delta_B\lambda_B$, respectively. Here $q$ is the common-cause parameter; $\delta_A=0.3$ and  $\delta_B=0.6$ are the sharing load parameters, while $\lambda_A=0.002$ and  $\lambda_B=0.01$ are the failure rates for components $A$ and $B$, respectively. Common-cause failure (when both parts fail) can occur with the intensity $q(\delta_A\lambda_A+\delta_B\lambda_B)$, and when this occurs, there are even chances that either part will be repair first. This differentiates the states $S_4$ and $S_5$ corresponding to the repair of  part $A$ and $B$ first, respectively. As a result, $\lambda_5=\lambda_6=q(\delta_A\lambda_A+\delta_B\lambda_B)/2$. 

After a part fails, the other part can fail as well with a regular failure rate for each part, so $\lambda_3=\lambda_B$ and  $\lambda_4=\lambda_A$. Finally, the repairs of each part have a fixed delay of $\tau=10$ (this last  parameter is changed as compared to ~\cite{Distefano2012}, so that  difference with the simple Markov approximation can be observed).
The Markov transition matrix has the following form:
\begin{equation}
\quad Q=\left(
\begin{array}{ccccc}
\cdot&\lambda_1&0&0&\hat{\mu}_4\\
\mu_1& \cdot&\mu_2&0&0\\
0&\lambda_1&\cdot&\hat{\mu}_3&0\\
\lambda_3&\lambda_6&0&\cdot&0\\
0&\lambda_5&\lambda_4&0&\cdot\\
\end{array}
\right)
\label{kolm33}
\end{equation} 
Here the diagonal terms are shown with $``\cdot"$ for brevity (they are calculated as shown in the right part of Eq.~\ref{kolm1}).
The states $S_1$ and $S_3$ are the renewal states, so Eq.~\ref{haza5} can be directly used to evaluate the repair rates from those states:
\begin{eqnarray}
\mu_1=\frac{\lambda_3}{e^{\lambda_3 \tau}-1},\quad \mu_2=\frac{\lambda_4}{e^{\lambda_4 \tau}-1}
\label{di1}
\end{eqnarray}
The repair rates for the states $S_4$ and $S_5$ are somewhat more challenging, as there are two distinct possibilities for each state: if the previous state was $S_2$ the repair has only started, and so the rates from Eq.~\ref{di1} are appropriate. On the other hand, a part might fail when it operates alone, so the repair of another part has already started (and so Eq.~\ref{mttr} provides the appropriate mean time to repair). Let us consider the state $S_4$ first. If the previous state was $S_3$, then the corresponding failure rate (using Eq.~\ref{mttr}) is 

\begin{eqnarray}
\mu_3=\frac{\lambda_3F_3(\tau)}{\tau\lambda_3-F_3(\tau)}
\label{di2}
\end{eqnarray}

The total repair rate for  $S_3$ is a weighted average of the two rates, where the rates are based on the relative intensities of the inflows to the state:

\begin{eqnarray}
\hat{\mu}_3=\frac{P_1 \lambda_3\mu_3+P_2 \lambda_6\mu_1}{P_1 \lambda_3+P_2 \lambda_6}
\label{di3}
\end{eqnarray}

Similarly, for $S_3$ we obtain:

\begin{eqnarray}
\mu_4=\frac{\lambda_4F_4(\tau)}{\tau\lambda_4-F_4(\tau)}, \quad \hat{\mu}_4=\frac{P_3 \lambda_4\mu_4+P_2 \lambda_5\mu_2}{P_3 \lambda_4+P_2 \lambda_5}
\label{di4}
\end{eqnarray}

Finally, we can apply Lemma 3, which stipulates that $\hat{\mu}_3$ and $\hat{\mu}_4$ don't affect the relative probabilities for states $S_1$, $S_2$, and $S_3$. Therefore, a simple procedure allows us to find the steady-state solution: first we solve the system of equations $Q\mathbf{P}=0$ (as usual for finding the steady-state  solutions, one of the rows in Eq.~\ref{kolm33} is replaced with the normalization condition for  with any values of  $\hat{\mu}_3$ and $\hat{\mu}_4$), then we use the probabilities $P_1$, $P_2$, and $P_3$ in Eqs~\ref{di3}, \ref{di4} to calculate the actual values for $\hat{\mu}_3$ and $\hat{\mu}_4$, and solve the linear system again (linear solvers are so computationally efficient  that partitioning the matrix, in order to solve the problem in a single step, is not worthwhile).

Figure~\ref{distefano} shows the chances of the down states ({\em i.e.,} $P_4+P_5$) as a function of the common-cause parameter $q$ for the full model as well as for Markov approximation.

\begin{figure}[!t]
\centering
\includegraphics[width=3in]{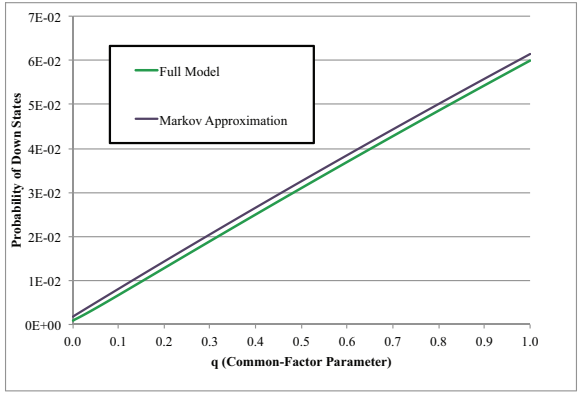}
\caption{Sensitivity of the chances of the down state as a function of the common-cause parameter $q$}
\label{distefano}
\end{figure}

\section{Conclusion}
A simple practical procedure for correcting Markov models for non-Markovian effects is presented.  In contrast to the existing approaches, the focus is on the asymptotic analysis of the inflows and outflows of probabilities for individual states.  As a result, direct impact of the choice of the distributions for particular transitions can be locally assessed by converting them into equivalent exponential transitions. For a steady-state analysis of a semi-Markovian process, the result is equivalent to the two-step procedure involving an embedded discrete Markov chain. The local nature of the analysis presented in the current paper is suitable for large-scale applications:  the correction to the transition rates can be conducted only when needed as pre-processing, before assembling the global Markov model. 

The inflows-outflows perspective for individual states has been further utilized to evaluate the asymptotic system transition rate to an absorbing state (hazard rate) for both Markovian and semi-Markovian properties. Several general properties for this hazard rate have been established, facilitating the scalability of the developed solutions.  Finally,  the same local approach has been applied to tackle two relatively simple examples that exhibit essentially non-Markovian effects (the holding times from the previous system states must be taken into account). The resulting  compactness of the solutions provides hope that general procedures for tackling large-scale problems with non-Markovian effects  can be developed in the future.

\ifCLASSOPTIONcaptionsoff
  \newpage
\fi



\bibliographystyle{IEEEtran}
\bibliography{delmain}

\begin{thebibliography}{10}
\providecommand{\url}[1]{#1}
\csname url@samestyle\endcsname
\providecommand{\newblock}{\relax}
\providecommand{\bibinfo}[2]{#2}
\providecommand{\BIBentrySTDinterwordspacing}{\spaceskip=0pt\relax}
\providecommand{\BIBentryALTinterwordstretchfactor}{4}
\providecommand{\BIBentryALTinterwordspacing}{\spaceskip=\fontdimen2\font plus
\BIBentryALTinterwordstretchfactor\fontdimen3\font minus
  \fontdimen4\font\relax}
\providecommand{\BIBforeignlanguage}[2]{{%
\expandafter\ifx\csname l@#1\endcsname\relax
\typeout{** WARNING: IEEEtran.bst: No hyphenation pattern has been}%
\typeout{** loaded for the language `#1'. Using the pattern for}%
\typeout{** the default language instead.}%
\else
\language=\csname l@#1\endcsname
\fi
#2}}
\providecommand{\BIBdecl}{\relax}
\BIBdecl

\bibitem{Distefano2012}
S.~Distefano and K.~S. Trivedi, ``Non-{M}arkovian state-space models in
  dependability evaluation,'' \emph{Quality and Reliability Engineering
  International}, vol.~29, pp. 225--239, 2012.

\bibitem{Doorn2013}
E.~A. van Doorn and P.~K. Pollett, ``Quasi-stationary distributions for
  discrete-state models,'' \emph{European Journal of Operational Research},
  vol. 230, pp. 1--14, 2013.

\bibitem{Lamb1895}
H.~Lamb, \emph{Hydrodynamics}.\hskip 1em plus 0.5em minus 0.4em\relax Cambridge
  at the University Press, 1895.

\bibitem{LandauLifshitz1987}
L.~D. Landau and E.~Lifshitz, \emph{Fluid Mechanics}, 2nd~ed., ser. Course of
  Theoretical Physics.\hskip 1em plus 0.5em minus 0.4em\relax
  Butterworth-Heinemann, 1987, vol.~6.

\bibitem{Marsan1990}
M.~A. Marsan, ``Stochastic {P}etri nets: An elementary introduction,'' in
  \emph{Advances in {P}etri Nets 1989}, ser. Lecture Notes in Computer Science,
  G.~Rozenberg, Ed.\hskip 1em plus 0.5em minus 0.4em\relax Springer, 1990, vol.
  424, pp. 1--29.

\bibitem{Haas2002}
P.~J. Haas, \emph{Stochastic {P}etri Nets. Modelling, Stability,
  Simulation}.\hskip 1em plus 0.5em minus 0.4em\relax New York: Springer, 2002.

\bibitem{Volovoi2015}
V.~Volovoi, ``Tutorial: Simulation with stochastic {P}etri nets,'' in
  \emph{Winter Simulation Conference}, December, 6--9 2015.

\bibitem{Yaglom1947}
A.~M. Yaglom, ``Certain limit theorems of the theory of branching processes,''
  \emph{Doklady Akademii Nauk SSSR}, vol.~56, pp. 795--798 (in Russian), 1947.

\bibitem{Volovoi2017mmr}
V.~Volovoi, ``Deferred maintenance of redundant systems,'' in
  \emph{Mathematical Methods in Reliability}, vol. 10th, Grenoble, France,
  July, 3-6 2017.

\bibitem{Darroch1965}
J.~Darroch and E.~Seneta, ``On quasi-stationary distributions in absorbing
  discrete-time finite {M}arkov chains,'' \emph{Journal of Applied
  Probability}, vol.~2, pp. 88--100, 1965.

\bibitem{Boussemart2001}
M.~Boussemart, T.~Bickard, and N.~Limnios, ``Markov decision processes with a
  constraint on the asymptotic failure rate,'' \emph{Methodology and Computing
  in Applied Probability}, vol.~3, no.~2, pp. 199 -- 214, 2001.

\bibitem{Darroch1967}
J.~Darroch and E.~Seneta, ``On quasi-stationary distributions in absorbing
  continuous-time finite {M}arkov chains,'' \emph{Journal of Applied
  Probability}, vol.~4, pp. 192--196, 1967.

\bibitem{Haverkort2001}
B.~R. Haverkort, ``Markovian models for performance and dependability
  evaluation,'' in \emph{Lectures on Formal Methods and Performance Analysis},
  ser. Lecture Notes in Computer Science, E.~Brinksma, H.~Hermannsa, and J.-P.
  Katoen, Eds.\hskip 1em plus 0.5em minus 0.4em\relax Springer, 2001, vol.
  2090, pp. 38--83.

\bibitem{Boyd2004}
S.~Boyd, P.~Diaconis, and L.~Xiao, ``Fastest mixing {M}arkov chain on a
  graph,'' \emph{{SIAM} Review}, vol.~46, no.~4, pp. 667--689, 2004.

\bibitem{VolovoiVega2012}
V.~Volovoi and R.~V. Vega, ``On compact modeling of coupling effects in
  maintenance processes of complex systems,'' \emph{International Journal of
  Engineering Science}, vol.~51, pp. 193--210, 2012.

\bibitem{Whitt1993}
W.~Whitt, ``Approximations for the {GI/G/m} queue,'' \emph{Production and
  Operations Management}, vol.~2, no.~2, pp. 114--160, 1993.

\bibitem{Pijnenburg1993}
M.~Pijnenburg, N.~Ravichandran, and G.~Regterschot, ``Stochastic analysis of a
  dependent parallel system,'' \emph{European Journal of Operational Research},
  vol.~68, pp. 90--104, 1993.

\end{thebibliography}
\end{document}